\documentclass[aps,prx,10pt,twocolumn,amsmath,amssymb,floatfix,superscriptaddress,longbibliography]{revtex4-2}
\pdfoutput=1
\usepackage[outline]{contour}

\usepackage[T1]{fontenc}
\usepackage{color}
\usepackage{bbm}
\usepackage{graphicx}
\usepackage{dcolumn}
\usepackage[utf8]{inputenc}
\usepackage{graphicx}
\usepackage{epstopdf}
\usepackage{amsthm}
\usepackage{amsmath}
\usepackage{empheq}
\usepackage{eufrak}
\usepackage{bbm}
\usepackage{braket}
\usepackage{amssymb}
\usepackage{mathtools}
\usepackage[usenames,dvipsnames]{xcolor}
\usepackage{tikz}
\usepgflibrary{shapes.arrows}
\usetikzlibrary{calc, positioning, shapes.arrows}
\usepackage{cases}
\usepackage{smartdiagram}
\usepackage{latexsym}
\usepackage[colorlinks=true,citecolor=Cerulean,linkcolor=RubineRed,urlcolor=Cerulean]{hyperref}
\usepackage{cleveref}

\contourlength{0.5pt}
\contournumber{20}

\definecolor{S_Blue}{RGB}{0,135,252}
\definecolor{S_Red}{RGB}{214,13,63}
\definecolor{Blue}{RGB}{47,89,151}
\definecolor{S_Grey}{RGB}{150,150,158}
\definecolor{S_Yel}{RGB}{255,204,0}
\definecolor{S_Green}{RGB}{102,204,0}
\definecolor{S_Brown}{RGB}{154,41,41}

\DeclareFontFamily{OT1}{pzc}{}
\DeclareFontShape{OT1}{pzc}{m}{it}{<-> s * [1.15] pzcmi7t}{}
\DeclareMathAlphabet{\mathpzc}{OT1}{pzc}{m}{it}

\graphicspath{ {figures/figure1/}{figures/figure2/} }

\newcommand{\tr}[1]{\operatorname{\textnormal{Tr}}\left( {#1} \right)} 

\theoremstyle{plain}
\newtheorem{theorem}{Theorem}

\newtheorem{lemma}{Lemma}
\newtheorem{corollary}{Corollary}

\newcommand{\rvline}{\hspace*{-\arraycolsep}\vline\hspace*{-\arraycolsep}}

\newcommand*{\SavedEqref}{}
\let\SavedEqref\eqref
\renewcommand*{\eqref}[1]{%
  \begingroup
    \hypersetup{
      linkcolor=CornflowerBlue,
      linkbordercolor=CornflowerBlue,
    }%
    \SavedEqref{#1}%
  \endgroup
}

\begin{document}

\title{Lower bounds on the number of rounds of the quantum approximate optimization algorithm required for guaranteed approximation ratios}

\author{Naphan Benchasattabuse}
\email{whit3z@sfc.wide.ad.jp}
\affiliation{CCS-3 Information Sciences, Los Alamos National Laboratory, Los Alamos, NM 87545, USA}
\affiliation{Graduate School of Media and Governance, Keio University, 5322 Endo, Fujisawa, Kanagawa 252-0882, Japan}
%
\author{Andreas B\"{a}rtschi}
\email{baertschi@lanl.gov}
\affiliation{CCS-3 Information Sciences, Los Alamos National Laboratory, Los Alamos, NM 87545, USA}
%
\author{Luis Pedro Garc\'{i}a-Pintos}
\email{lpgp@lanl.gov}
\affiliation{T-4 Quantum and Condensed Matter Physics, Los Alamos National Laboratory, Los Alamos, NM 87545, USA}
%
\author{John Golden}
\affiliation{CCS-3 Information Sciences, Los Alamos National Laboratory, Los Alamos, NM 87545, USA}
%
\author{Nathan Lemons}
\affiliation{T-5 Applied Mathematics and Plasma Physics, Los Alamos National Laboratory, Los Alamos, NM 87545, USA}
%
\author{Stephan Eidenbenz}
\email{eidenben@lanl.gov}
\affiliation{CCS-3 Information Sciences, Los Alamos National Laboratory, Los Alamos, NM 87545, USA}

\date{\today}

\begin{abstract}
The quantum approximate optimization algorithm, also known in its generalization as the quantum alternating operator ansatz, (QAOA) is a heuristic hybrid quantum-classical algorithm for finding high-quality approximate solutions to combinatorial optimization problems, such as maximum satisfiability. 
While the QAOA is well studied, theoretical results as to its runtime or approximation ratio guarantees are still relatively sparse.
We provide some of the first lower bounds for the number of rounds (the dominant component of QAOA runtimes) required for the QAOA.
For our main result, we (i) leverage a connection between quantum annealing times and the angles of the QAOA to derive a lower bound on the number of rounds of the QAOA with respect to the guaranteed approximation ratio. 
We apply and calculate this bound with Grover-style mixing unitaries and (ii) show that this type of QAOA requires at least a polynomial number of rounds to guarantee any constant approximation ratios for most problems.
We also (iii) show that the bound depends only on the statistical values of the objective functions, and when the problem can be modeled as a $k$-local Hamiltonian, can be easily estimated from the coefficients of the Hamiltonians.
For the conventional transverse-field mixer, (iv) our framework gives a trivial lower bound to all bounded-occurrence local cost problems and for all strictly $k$-local cost Hamiltonians matching known results that constant approximation ratio is obtainable with a constant-round QAOA for a few optimization problems from these classes.
Using our proof framework, (v) we recover the Grover lower bound for unstructured search and, with small modification, show that our bound applies to any QAOA-style search protocol that starts in the ground state of the mixing unitaries.
\end{abstract}

\maketitle

\section{Introduction}
\label{sec:introductions}

Solving optimization problems has been a cornerstone computational task, with a wide range of applications across industry, science, and technology. 
The two leading approaches in tackling optimization problems in quantum computing are the quantum annealing (QA)~\cite{kadowaki-nishimori-annealing} and the quantum approximate optimization algorithm~\cite{farhi-qaoa} and its generalization to constrained optimization problem, the quantum alternating operator ansatz algorithm (QAOA)~\cite{hadfield-qaoa}.
In order to solve an optimization problem on a quantum computer, it is natural to translate the minima or maxima finding of a given objective function to finding the ground state of a corresponding problem Hamiltonian. 
Although the QAOA does not have any proven theoretical quantum advantage or speedups over classical algorithms, several numerical simulation results hint at potential ``practical'' speedups if and when scalable quantum computers exist~\cite{john-numerical-speedup-constrained, john-sat-qaoa, boulebnane-montanaro-sat-qaoa}.
There are in fact algorithms for exact optimization with provable polynomial speedups over random guessing (brute force) with Grover-style adaptive search~\cite{durr-hoyer-minimum-finding, gilliam-grover-adaptive-search-cpbo}, but this type of algorithm does not provide near-optimal solutions or any approximation ratio guarantees.
In its simplest form, the QAOA consists of two parametrized subcircuits, the phase separator that encodes the problem Hamiltonian $H_1$ and the mixing operator (or the mixer) that represents a Hamiltonian $H_0$ that does not commute with the phase separator.
The two subcircuits are then applied alternately in rounds, the number of which is usually denoted by $p$, and the goal is to adjust the parameters to minimize the expectation value to the given problem Hamiltonian $\braket{H_1}$.

For the QA protocol, the system is initialized in an easy-to-prepare ground state of a Hamiltonian $H_0$ that does not commute with the problem Hamiltonian $H_1$, similar to the mixer of the QAOA.
The system is then driven by a time-dependent Hamiltonian, $H(t) = [1 - g(t)] H_0 + g(t) H_1$, where $g(t)$ is called the annealing schedule, where $g(0) = 0$ and $g(t_f) = 1$, with $t_f$ the total evolution time or the annealing time.
It is known from the adiabatic theorem~\cite{born-fock-adiabatic-theorem} that if the transition from $H_0$ to $H_1$ is slow enough, the state will remain close to the instantaneous ground state at all times and will arrive at the ground state of $H_1$ at the end of the protocol.
Note that although the adiabatic theorem guarantees that the final state will be the ground state of $H_1$, the timescale of the adiabatic annealing time is polynomial in the inverse of the gap between the ground state and the first excited state.
This gap is typically exponentially small, leading to an exponential time scale in the number of qubits, and thus also in the problem size.
Thus nonadiabatic annealing schedules, such as counterdiabatic driving~\cite{demirplak-rice-counterdiabatic-annealing, berry-transitionless-driving, guery-odelin-shortcuts-to-adiabaticity, del-campo-shortcuts-to-adiabaticity-counterdiabatic}, diabatic annealing~\cite{crosson-lidar-diabatic-annealing}, or the optimal control approach~\cite{guery-odelin-shortcuts-to-adiabaticity}, have been proposed that still fall into the framework of quantum annealing.

It is straightforward to see that in the limit $p \rightarrow \infty$, the QAOA is the Trotterization of the QA.
Since $p$, the number of QAOA rounds corresponding to the circuit depth, relates directly to the time complexity, it would be of interest to characterize the asymptotic scaling of $p$ with the solution quality and the problem size.
In the case of finite $p$, the intuition is not as clear, as the Trotterization argument would need to take the Trotter error into account.
On the other hand, the QAOA does not need to represent the Trotterization or discretization of any QA annealing schedule and can be thought of purely as an ansatz structure for variational algorithms~\cite{cerezo-vqa-review}.
In the finite-$p$ regime, it has been proven that certain problems can admit a fixed solution quality~\cite{farhi-qaoa, farhi-e3lin2, wurtz-love-max-cut-constant-depth-qaoa}.
Another aspect that is widely studied for the variational quantum algorithms is the expressibility of the QAOA circuits~\cite{akshay-reachability-deficit-prl}, how much of the solution space can be reached with a given number of $p$ rounds of the QAOA~\cite{akshay-circuit-depth-scaling-sat-density, niu-chuang-qaoa-state-transfer}, or the concentration of states~\cite{basso-constant-depth-qaoa-analysis, anshu-metger-concentration-bound-qaoa}.
Other connections between QA and the QAOA that have been explored are the initial-state selection~\cite{cain-qaoa-get-stuck, sack-serbyn-qaoa-initialization-via-qa, sack-greedy-qaoa-initialization, he-ruslan-alignment-qaoa} and the counterdiabatic driving in the QAOA~\cite{wurtz-love-counterdiabaticity-qaoa}.
In the other direction, QA can also simulate QAOA protocols by realizing that the alternate Hamiltonian applications of the QAOA correspond to the bang-bang schedule~\cite{innocenti-bang-bang}, i.e., the schedule where only one driving Hamiltonian is turned on at a time.
It is also conjectured that the bang-bang schedule is optimal for QA for most problems~\cite{yang-chamon-bang-bang-optimal-vqa} while some other problems might have the optimal schedule that follows the bang-anneal-bang protocol~\cite{brady-optimal-controls-bang-anneal-bang}.

In this work, we give lower bounds on the number of QAOA rounds needed to achieve a constant target approximation ratio via results from the lower bounds of annealing time derived for QA~\cite{luis-pedro-annealing-lower-bounds}.
The structure of the paper is as follows.
We begin by defining the QAOA and QA and how they work in our framework.
We outline the combinatorial optimization problems and their objective functions we consider, maximization problems with only non-negative integer values, and briefly recap the lower bound results of QA from~\cite{luis-pedro-annealing-lower-bounds} (Sec.~\ref{sec:prelim}).
We then outline the assumptions of our driving Hamiltonians and derive the lower bounds for the number of rounds $p$ in the QAOA case (Sec.~\ref{sec:main-results}).
We also show that rescaling of the Hamiltonians has no observable effect on the QAOA, which at first, sounds counterintuitive when compared to QA, and show that the scaling factor can be selected to make the bounds constantly tighter without affecting other assumptions we prior made (Sec.~\ref{sec:hamiltonian-rescaling}).
After we have the bounds established, we show that when using the Grover diffusion operator style as the mixer, the bounds depend on statistical values of the objective values (Sec.~\ref{sec:grover-mixer}).
Later we look into the most common mixer, the transverse field.
We describe that although the lower bounds for certain families of $k$-local Hamiltonians (bounded occurrence or strictly $k$-local) evaluate to trivial lower bounds, the results agree with the constant-depth QAOA literature (Sec.~\ref{sec:tf-mixer}).
Next we look into solving search problems with the QAOA and how our lower bounds behave in this special case.
We later show that a slight modification to our lower bound can recover the oracular unstructured search lower bound for both the Grover and the transverse-field mixer and also generalize to cover the continuous-time quantum walk search (Sec.~\ref{sec:search-lower-bound}).
We then end the paper with an outlook on the implications and impact on our derived bounds and how future research directions can be taken from further investigation into other properties of the system such as the coherence and the expectation value changes over the course of the protocol (Sec.~\ref{sec:discussion}).

\section{Background}
\label{sec:prelim}

In this section, we give a brief introduction to combinatorial optimization, the QAOA, and QA. 
We then outline our assumptions and constraints on the optimization problems that we consider leading to constraints on problem Hamiltonians which are used to drive the QAOA.
We give a summary of results on lower bounds on annealing time derived for QA that we later use to derive the lower bounds on the number of QAOA rounds.
We also elaborate on the connection between the QAOA and QA with a bang-bang schedule.

\subsection{Combinatorial optimization problem}
\label{sec:prelim:optimization-problem}

Here we consider combinatorial optimization problems defined on bit strings of length $n$ with an objective function that maps bit string $x$ to non-negative integer $C(x)$, where we wish to find the maximum objective value \mbox{$C_{\text{max}} = \max_z C(z)$}.
It is also useful to define the average of the objective values \mbox{$C_{\text{avg}} = \frac{1}{N} \sum_z C(z)$} and the standard deviation \mbox{$\sigma_C = \sqrt{\sum_z \frac{[C(z) - C_{\text{avg}}]^2}{N}}$}.
In this work, we only consider the case where all objective values are non-negative integers [$C(x) \in \mathbbm{N}_0$ for all $x$]. 
We use $F \subseteq \{0, 1\}^n$ to refer to the set of all feasible solution strings.
We also often refer to the size of the search space $N = |F|$, when the problem is unconstrained, the size of the feasible search space becomes $N = 2^n$.
When the context of the feasible solution set is clear or not of importance, we will refer to its size simply as $N$.
From the objective function, we define the cost Hamiltonian
\begin{align}
    H_C = \sum_{z \in F} C(z) \ket{z}\bra{z};
\end{align}
thus $\ket{z}$ is an eigenstate in the computational basis $Z^{\otimes n}$.
One of the most common ways to evaluate whether a heuristic optimization algorithm is good or bad is to determine its \emph{approximation ratio}.
In this work, we define the approximation ratio to be $\lambda \in [0, 1]$ by
\begin{align}
    \lambda = \frac{\braket{H_C}_p}{C_{\text{max}}} = \max_{\ket{\psi_p}} \frac{\braket{\psi_{p} | H_C | \psi_p}}{C_{\text{max}}},
\label{eq:approximation-ratio}
\end{align}
where $\ket{\psi_p}$ is the state obtainable from $p$ rounds of the QAOA and $\braket{H_C}_p$ is the largest expectation value that can be obtained from $p$ rounds of the QAOA with respect to the cost Hamiltonian (or the objective function).

Since it is more common to talk about Hamiltonians with zero ground-state energies, unless otherwise specified, we define our problem Hamiltonian to be
\begin{align}
    H_1 = \mathbbm{1} C_{\text{max}} - H_C.
\label{eq:h1}
\end{align}
We note that the solution to our optimization problem is now mapped to the ground state of the problem Hamiltonian $H_1$, turning the maximization problem into minimization in the problem Hamiltonian.

\subsection{QAOA}
\label{sec:prelim:qaoa}

A QAOA protocol consists of two non-commuting Hamiltonians, the problem encoding Hamiltonian $H_1$ (eq.~\eqref{eq:h1}) and a mixing Hamiltonian with zero ground-state energies $H_0$, a hyperparameter $p$ indicating how many rounds the two Hamiltonians are alternately applied, and the set of angles $\vec{\gamma}$ and $\vec{\beta}$ denoting the evolving time of $H_1$ and $H_0$ at each round.
The state after $p$ rounds of the QAOA can be defined as 
\begin{widetext}
\begin{equation}
\label{eq:qaoa-circuit}
    \ket{\psi (\vec{\beta}, \vec{\gamma})}_p = (e^{-i \beta_p H_0} e^{-i \gamma_p H_1}) (e^{-i \beta_{p-1} H_0} e^{-i \gamma_{p-1} H_1}) \cdots (e^{-i \beta_1 H_0} e^{-i \gamma_1 H_1})  \ket{\psi_0}.
\end{equation}
\end{widetext}

The goal of the QAOA is to prepare a state that minimizes $\braket{\psi (\vec{\beta}, \vec{\gamma})| H_1 | \psi (\vec{\beta}, \vec{\gamma})}_p$, which is done by alternating between sampling the states prepared by the QAOA circuit and then on classical computers evaluating the expectation value and adjusting the parameter vectors $\vec{\gamma} = (\gamma_1, \gamma_2, \hdots, \gamma_p)$ and $\vec{\beta} = (\beta_1, \beta_2, \hdots, \beta_p)$ so the expectation value moves toward the minima.
We will later refer to $\braket{\psi (\vec{\beta}, \vec{\gamma})| H_1 | \psi (\vec{\beta}, \vec{\gamma})}_p$ with the shorthand notation $\braket{H_1}_p$.

\subsection{Quantum Annealing}
\label{sec:prelim:annealing}

Quantum annealing protocols consist of a set of driving Hamiltonians and an annealing schedule, a function of time $g(t)$.
The set of driving Hamiltonians usually consists of two Hamiltonians, the problem Hamiltonian $H_1$ whose ground state we want to find and a Hamiltonian $H_0$ whose ground state is known and easy to prepare.
The system is then driven by a time-dependent Hamiltonian given by
\begin{equation}
\label{eq:annealing-schedule}
    H(t) = [1 - g(t)]H_0 + g(t) H_1,
\end{equation}
where the function $g(t)$ is called the annealing schedule that satisfies $g(0) = 0$ and $g({t_f}) = 1$, with $t_f$ the total duration of the process. 
Here we consider an annealing schedule to be successful if the probability of finding the ground state of the problem Hamiltonian $p_{0, t_f} \geq k$ where $k$ is some constant threshold.

\subsection{QAOA as a bang-bang QA}
\label{sec:prelim:qaoa-qa-connection}

Given a time-independent Hamiltonian $H$, we can write down the evolution of some initial quantum state $\ket{\psi(0)}$ with time $t$ under this Hamiltonian as
\begin{align}
    \ket{\psi(t)} = e^{-iHt} \ket{\psi(0)},
\end{align}
where $\ket{\psi(t)}$ is the state at time $t$.
From this perspective, what the QAOA does in eq.~\eqref{eq:qaoa-circuit} at each round $j$ is evolve the state by time-independent Hamiltonians $H_1$ and $H_0$ with times $\gamma_j$ and $\beta_j$, respectively.

When the driving Hamiltonians of the QAOA and QA are the same, the QAOA is essentially the bang-bang schedule, the schedule in which $g(t)$ only has two distinct values, $0$ and $1$, at all times.
The annealing time can then be defined as
\begin{align}
    t_\text{anneal} = \sum_{j=1}^p \left( |\beta_j| +  |\gamma_j| \right).
\label{eq:sum-of-angles-to-annealing-time}
\end{align}
In this view, any result that applies to general QA can also be applied to the QAOA.

\subsection{Lower bounds on annealing times}
\label{sec:prelim:lower-bound-annealing-times}

It was shown in~\cite{luis-pedro-annealing-lower-bounds} that a lower bound on the annealing time to reach the ground state with high probability can be derived independently of the trajectory of the state through the process.

\begin{theorem}[Garc\'{i}a-Pintos et al.~\cite{luis-pedro-annealing-lower-bounds}]
Given two driving Hamiltonians with zero ground-state energies $H_0$ and $H_1$,
for a quantum annealing protocol that starts in the ground state of $H_0$ and ends at time $t_{f}$,
it holds that
\begin{align}
    t_f \geq \tau_\text{anneal} \coloneqq \frac{\braket{H_0}_{t_f} + \braket{H_1}_0 - \braket{H_1}_{t_f}}{\|[H_1, H_0]\|}.
\end{align}
\label{theorem:annealing-lower-bound}
\end{theorem}

Here $\braket{H}_t = \braket{\psi_t | H | \psi_t}$ is the expectation value of the evolving state $\ket{\psi_t}$ at time $t$ with respect to the Hamiltonian $H$ and $\| \cdot \|$ denotes the spectral norm.
This bound only depends on the properties of the problem Hamiltonian and the driving Hamiltonians, with the term $\braket{H_1}_{t_f}$ capturing the error to the ground state of $H_1$.

We note that there are also other variations of the lower bounds that are tighter.
For those lower bounds, which take the trajectory of the state during the annealing process, we refer readers to~\cite{luis-pedro-annealing-lower-bounds}.

\section{Lower bounds on QAOA}
\label{sec:main-results}
\label{sec:main-results:angle-lower-bound}

In this section, we give the main results, which show the relationship between the approximation ratio $\lambda$ and the number of rounds $p$ for the QAOA required to reach it.
Let $\braket{H}_l$ denote the expectation value after $l$ rounds of the QAOA.
(In the context of QA $l$ refers to the annealing time, while for the QAOA $l$ refers to the state after $l$ rounds of the QAOA.)

\subsection{Lower bounds on QAOA angles}

From the view that the QAOA is a QA with a bang-bang schedule (Sec.~\ref{sec:prelim:qaoa-qa-connection}), we can derive lower bounds on the QAOA given the angles.

\begin{lemma}
Given a phase separator Hamiltonian $H_1$ and a mixing Hamiltonian $H_0$, where both Hamiltonians have zero ground-state energies,
for a QAOA protocol with $p$ rounds driven by $H_0$ and $H_1$ with angle parameters $\vec{\gamma}$ and $\vec{\beta}$ that starts from the ground state of $H_0$, then
\begin{align}
    \sum_{j=1}^p \left( |\beta_j| +  |\gamma_j| \right) \geq \frac{\braket{H_0}_{p} + \braket{H_1}_0 - \braket{H_1}_{p}}{\|[H_1, H_0]\|}.
\end{align}
\label{lemma:qaoa-ground-state-bound}
\end{lemma}

\begin{proof}
This is an immediate result from Theorem~\ref{theorem:annealing-lower-bound} and eq.~\eqref{eq:sum-of-angles-to-annealing-time}.
\end{proof}

It should be noted that the assumption that the driving Hamiltonians $H_0$ and $H_1$ have zero ground-state energies does not violate the usual construction of the QAOA.
One can always shift the eigenvalues of a Hamiltonian by a constant.
From the perspective of optimization, adding a constant cost to the objective function preserves the solution quality.
The unitaries constructed from the modified Hamiltonians would only differ by a global phase.

\subsection{Lower bounds on QAOA rounds}
\label{sec:main-results:rounds-lower-bound}

Before we go on to our main results, we will add one more assumption to the driving Hamiltonians.
In addition to having zero ground-state energies, we also assume that the driving Hamiltonians are periodic.
Formally, 
\begin{align}
    e^{iH(t + \theta)} = e^{iHt},
\end{align}
where $\theta$ is the period of the Hamiltonian $H$.
For $H_1$, the problem Hamiltonian, the periodicity of $2\pi$ comes naturally since we restrict ourselves to only consider combinatorial optimization with non-negative integer objective values.
For $H_0$, the mixer, most mixing Hamiltonians used in the literature are also periodic, with only a few exceptions like the ring mixer for some fixed nearby values~\cite{hadfield-qaoa}.

Now we are ready to give the relationship between the approximation ratio $\lambda = \braket{H_C}_p/C_{\text{max}}$ and the lower bound on the number of QAOA rounds required to reach this expectation value.
\begin{theorem}
Given a classical objective function $C(x)$ for a maximization task, represented by the Hamiltonian $H_C$, encoded into the phase separator Hamiltonian $H_1 = \mathbbm{1}C_{\text{max}} - H_C$ and a mixing Hamiltonian $H_0$, where all Hamiltonians are $2\pi$ periodic with zero ground-state energies and letting $C_{\text{max}}$ and $C_{\text{avg}}$ denote the global maximum and the average of $C(x)$, respectively, if a QAOA protocol with $p$ rounds driven by $H_0$ and $H_1$ that starts from the ground state $\ket{\psi_0} = \ket{+}^{\otimes n}$ of $H_0$ reaches a state with approximation ratio $\lambda$, then
\begin{align}
p &\geq \frac{\braket{H_0}_{p} + \lambda C_{\text{max}} - C_{\text{avg}}}{4\pi \|[H_C, H_0]\|}.
\label{eq:qaoa-obj-phase-lower-bound}
\end{align}
\label{theorem:qaoa-objective-bound}
\end{theorem}

\begin{proof}
From Lemma~\ref{lemma:qaoa-ground-state-bound}, when imposing the periodicity constraints, we get
\begin{align}
    \sum_{j=1}^p \left( |\beta_j| +  |\gamma_j| \right) 
    &\geq \frac{\braket{H_0}_{p} + \braket{H_1}_0 - \braket{H_1}_{p}}{\|[H_1, H_0]\|} \\
p(2\pi + 2\pi)     
    &\geq \frac{\braket{H_0}_{p} + \braket{H_1}_0 - \braket{H_1}_{p}}{\|[H_1, H_0]\|} \label{eq:theorem-2-line-2}\\
p
    &\geq \frac{\braket{H_0}_{p} + \braket{H_1}_0 - \braket{H_1}_{p}}{4\pi \|[H_1, H_0]\|},
\end{align}
where \eqref{eq:theorem-2-line-2} comes from the fact that the two Hamiltonians are $2\pi$ periodic and thus the largest meaningful $\beta_j$ and $\gamma_j$ are bounded by $2\pi$.
Since the phase separator $H_1$ of the QAOA encodes the information of the objective value into the phase, we get
\begin{align}
\braket{H_1}_0 &= C_{\text{max}} - C_{\text{avg}}, \\
\braket{H_1}_p &= C_{\text{max}} - \lambda C_{\text{max}},\\
\braket{H_1}_0 - \braket{H_1}_p 
    &= \lambda C_{\text{max}} - C_{\text{avg}}.
\end{align}
This comes directly from the definitions of the approximation ratio [eq.~\eqref{eq:approximation-ratio}] and $H_1$ [eq.~\eqref{eq:h1}].
As for the denominator, we know that 
\begin{align}
    [A + kI, B] = [A, B]
\end{align}
holds for any constant $k$ and any matrix $A, B$. 
Because $H_1 = \mathbbm{1}C_{\text{max}} - H_C$, we get $\|[H_1, H_0]\| = \|[H_C, H_0]\|$.
\end{proof}

The above theorem also applies to constrained optimization problems, when the initial state $\ket{\psi_0}$ is the uniform superposition of all feasible states and the ground state of a constrained preserving mixing Hamiltonian $H_0$.

This result shows that the connection between the number of rounds $p$ and the approximation ratio $\lambda$ analytically applies to the general QAOA on combinatorial optimization problems independently of the choice of mixer or assumptions of the problem structure.
Although our lower bounds are not applicable to real-valued objective functions or nonperiodic mixers, it is a step forward to peek into the inner machinery and interplay between problem Hamiltonians and mixing Hamiltonians.

It should be pointed out that this lower bound also applies to families of QAOAs where the problem Hamiltonian $H_C$ is not in the set of driving Hamiltonians, e.g., the thresholded objective functions~\cite{john-threshold-qaoa}, as long as the chosen phase separator commutes with $H_C$.

\section{Rescaling of driving Hamiltonians}
\label{sec:hamiltonian-rescaling}

Earlier we made an assumption that the driving Hamiltonians must be periodic and have a period of $2 \pi$.
One could ask why this normalization of Hamiltonians makes sense.
One could also rescale all the driving Hamiltonians to have a similar range of energies, say, between zero and one, and this would be more physical, as is the case for QA.

\subsection{Changes in lower bounds}
\label{sec:hamiltonian-rescaling:lower-bounds}

Rescaling of the Hamiltonians changes their period and the range of expectation values.
Therefore, it impacts the lower bound in Theorem~\ref{theorem:qaoa-objective-bound}.
Suppose we rescale $H_0$ by $\alpha_0$ and $H_1$ by $\alpha_1$ where $\alpha_0, \alpha_1 > 0$; we get $H'_0 = \alpha_0 H_0$ and $H'_1 = \alpha_1 H_1$, and the new periods become $2\pi/\alpha_0$ and $2\pi/\alpha_1$ for $H'_0$ and $H'_1$, respectively.
The lower bound then becomes
\begin{align}
p \left(\frac{2 \pi}{\alpha_0} + \frac{2 \pi}{\alpha_1}\right) \geq \frac
    {\braket{H'_0}_p + \braket{H'_1}_0 - \braket{H'_1}_p}
    {\|[ H'_1, H'_0 ]\|} \\
p \geq \frac
    {\alpha_0 \braket{H_0}_p + \alpha_1 (\braket{H_1}_0 - \braket{H_1}_p)}
    {2\pi (\alpha_0 + \alpha_1) \|[ H_1, H_0 ]\|}.
\label{eq:qaoa-rescaling-lower-bound}
\end{align}
This allows us to analyze each Hamiltonian separately by scaling only one of the Hamiltonians to be so large that it dominates the numerator terms, for which we get two inequalities,
\begin{align}
p \geq \frac
    {\braket{H_0}_p}
    {2\pi \left|\left|\left[ H_1, H_0 \right]\right|\right|} \quad\text{ when } \alpha_0 \gg \alpha_1
\label{eq:lower-bound-h0-dominate}
\end{align}
and 
\begin{align}
p \geq \frac
    {\braket{H_1}_0 - \braket{H_1}_p}
    {2\pi \left|\left|\left[ H_1, H_0 \right]\right|\right|} \quad\text{ when } \alpha_0 \ll \alpha_1.
\label{eq:lower-bound-h1-dominate}
\end{align}
One consequence that we can directly observe is that the lower bounds cannot be adjusted to grow arbitrarily large via Hamiltonian rescaling. 
This is a desirable feature and leads to our next point.

\subsection{QAOA is invariant under Hamiltonian rescaling}
\label{sec:hamiltonian-rescaling:invariant}

If the driving Hamiltonians were to be rescaled for QA protocols, with $\alpha_j > 1$, one can realize that the time it takes to implement the same annealing schedule $g(t)$ would decrease.
If one were to have the ability to arbitrarily increase the strength with precise controls, then QA can be arbitrarily sped up.
We emphasize that this behavior does not apply to the QAOA, especially when our computational cost is the number of rounds. 
Since the effect of rescaling, as shown earlier, changes the period of the Hamiltonians, the angles $\beta$ and $\gamma$ need only be rescaled, and the number of rounds $p$ is not affected.

This implies that the lower bounds can be made tighter by adjusting $\alpha_0$ and $\alpha_1$ of eq.~\eqref{eq:qaoa-rescaling-lower-bound} to maximize the expression on the right-hand side.
Explicitly, we can write
\begin{align}
    p \geq \max_{\alpha_0, \alpha_1} \frac
     {\alpha_0 \braket{H_0}_p + \alpha_1 (\braket{H_1}_0 - \braket{H_1}_p)}
    {2\pi (\alpha_0 + \alpha_1) \|[ H_1, H_0 ]\|}.
\label{eq:worst-case-constant-bound}
\end{align}

We note that adjusting the lower bound via rescaling can only improve the constant factor and cannot improve the asymptotes compared to bounds derived from $2\pi$ period Hamiltonians.
These observations confirmed that deriving the lower bounds from the $2\pi$ period is logically sound even though it might seem counterintuitive at first when compared to the QA case.

One could ask why eq.~\eqref{eq:worst-case-constant-bound}, which represents the worst-case bound, should be used to lower bound $p$ since some other rescaling factors could make the lower bound go lower.
We emphasize again that although $\beta$ and $\gamma$ represent the evolution times of the Hamiltonians $H_0$ and $H_1$, they are rotation angles of parametrized unitaries; thus, worst-case scaling captures the lower bound for $p$.

Next we will use this bound that we derived and look at a few specific concrete examples.

\section{Lower bounds on QAOA rounds with Grover diffusion-style mixer}
\label{sec:grover-mixer}

One of the mixers that has been proposed for constrained optimization problems is the Grover mixer~\cite{andreas-grover-mixer-qaoa}, the parametrized version of the usual diffusion operator in Grover's algorithm and its amplitude amplification~\cite{grover-alg, brassard-bhmt-amplitude-amplification-estimation}.
This mixer not only preserves the state space in the space of feasible solutions, but is also invariant to permutations of the input since all states with the same objective value will always share the same amplitudes (fair sampling) if initialized with the same value.
This style of phase separator and Grover-mixer QAOA is also looked at from the perspective of amplitude amplification~\cite{satoh-spo, naphan-spo, koch-gaussian-amplitude-amplification, koch-variational-amplitude-amplification, shyamsundar-non-boolean-amplitude-amplification} in an attempt to decrease the cost of an adaptive Grover search and to gain more insights into generalization of ground-state finding via amplitude amplification techniques.
From the QAOA literature, this mixer was shown to provide better scaling for small problem size~\cite{akshay-reachability-deficit-prl, john-threshold-qaoa} than the traditional transverse field Hamiltonian due to its ability to mix rapidly~\cite{farhi-qaoa-whole-graph-typical-case, farhi-qaoa-whole-graph-worst-case, marsh-wang-optimization-with-quantum-walk-mixer, marsh-wang-quantum-walk-assisted-approximate-algorithm}.
Although this rapid mixing behavior has been empirically shown to be detrimental when the problem size is large when the classical outer loop is optimized via the expectation value~\cite{john-numerical-speedup-constrained, john-sat-qaoa}, we still find it interesting to characterize its lower bound since the optimal schedule could be difficult to find via optimizing over the expectation value~\cite{bittel-vqa-depth-qcma-hard}.

The Hamiltonian of the Grover mixer can be defined as
\begin{align}
H_\text{Grover} 
    &= \mathbbm{1} - \ket{\psi_0}\bra{\psi_0} \\
    &= \mathbbm{1} - \frac{1}{N}\sum_{x}\sum_{y} \ket{y}\bra{x},
\label{eq:grover-diffusion-hamiltonian}
\end{align}
where $x$ and $y$ are states in the feasible subspace.
Its unitary is the rotation around the uniform superposition state of all feasible states $\ket{\psi_0}$.

\begin{theorem}
Given a classical objective function $C(x)$ for a maximization task, represented by the Hamiltonian $H_C$, encoded into the phase separator Hamiltonian $H_1 = \mathbbm{1}C_{\text{max}} - H_C$ and the mixing Hamiltonian $H_\text{Grover}$, where all Hamiltonians are $2\pi$ periodic with zero ground-state energies, and letting $C_{\text{max}}$, $C_{\text{avg}}$, and $\sigma_{C}$ denote the global maximum, the average, and the standard deviation of $C(x)$, if a QAOA protocol with $p$ rounds driven by $H_\text{Grover}$ and $H_1$ starts from the ground state $\ket{\psi_0} = \ket{+}^{\otimes n}$ of $H_\text{Grover}$ and reaches a state with approximation ratio $\lambda$, then
\begin{align}
p
    &\geq \frac
    {1 - \left|\braket{\psi_0|\psi_{p}}\right|^2 + \lambda C_{\text{max}} - C_{\text{avg}}}
    { 4 \pi \sigma_{C} }.
\label{eq:grover-objective-bound}
\end{align}
\label{theorem:grover-objective-bound}
\end{theorem}

\begin{proof}
The term $\braket{H_\text{Grover}}_p$ can be evaluated directly from eq.~\eqref{eq:grover-diffusion-hamiltonian}.
Since $H_C$ is diagonal in the computational basis and $H_\text{Grover}$ is a rank-1 projector, the term $[H_C, H_\text{Grover}]$ evaluates to a projector of rank of at most 2.
There is a closed-form formula for calculating the spectral norm of a matrix of this form given in~\cite{ipek-eigenvalue-rank-1-skew-symmetric}, which gives (refer to Appendix~\ref{appendix:grover-obj-bound-derivations} for the full derivations)
\begin{align}
    \| [H_C, H_\text{Grover}] \| = \sigma_C.
\label{eq:spectral-norm-commutator-grover-mixer}
\end{align}
Plugging these values into Theorem~\ref{theorem:qaoa-objective-bound} gives the lower bound as shown.
\end{proof}

We can see that the above lower bound depends only on the statistical value of the combinatorial optimization problem that one wishes to solve.
(Evaluations with concrete problems are shown for Max-Cut in Sec.~\ref{sec:grover-mixer:approximation-ratio-unattainable} and for $k$-local cost Hamiltonians in the following section.)
This implies that the time required to solve an optimization problem using the QAOA with the Grover mixer depends only on the distribution of the objective values and not the structure of the problem, i.e., the closeness of good and bad solutions or whether there is a structure to be exploited from the bit string representations.
In other words, given two problem instances with the same distribution of objective values, the Grover mixer treats the two as equal and indifferent to the hardness of each instance.
This can be seen as a general result from the special case of single $\beta$ and $\gamma = \pi$ observed in~\cite{naphan-spo, koch-variational-amplitude-amplification} and agrees with observations made from angle selections in the infinite-size limit~\cite{headley-wilhelm-grover-qaoa-angle-finding}.

\subsection{Lower bounds on QAOA with Grover mixer on local cost Hamiltonian}
\label{sec:grover-mixer:k-local-hamiltonian}

So far we have treated both the objective function of our combinatorial optimization problems $C(x)$ and the corresponding cost Hamiltonian $H_C$ with minimal assumptions.
One natural assumption about most combinatorial optimization problems that can be efficiently implemented on quantum computers is that they need only local interactions between a small number of bits of the entire bit string~\cite{taillard-design-heuristic-algorithm-optimization-problem}.
This means that the corresponding $H_C$ is a classical $k$-local Hamiltonian (as examples shown in~\cite{hadfield-qaoa} with constrained search space and in~\cite{lucas-ising-formulation-np} with penalty terms) and can be written as
\begin{align}
H_C 
    &= \mathbbm{1}O(m) -\sum_{\nu=1}^m \alpha_\nu H_\nu \\
    &= \mathbbm{1}O(m) -\sum_{\nu=1}^m \alpha_\nu Z_{\nu,1} Z_{\nu,2} Z_{\nu,3} \cdots Z_{\nu,t},
\label{eq:k-local-hamiltonian}
\end{align}
where $m = \text{poly}(n)$, $\alpha_j$ is some real number of order $\Theta(1)$, and $t \leq k$ for all $\nu$ (each term involves at most $k$ qubits).
It should be noted that the locality here refers to the maximum number of qubit interactions for each cost term and not the spatial arrangement of the qubits on the quantum device.

In this form, one can see that the average of objective values $C_{\text{avg}}$ pops up naturally for unconstrained problems.
If we let $\ket{s} = \ket{+}^{\otimes n}$, we have
\begin{align}
    C_{\text{avg}} = \braket{s|H_C|s} \text{, } \braket{+|Z|+} = 0,
\end{align}
implying that the coefficient of the trivial Pauli term of eq.~\eqref{eq:k-local-hamiltonian} equals the average.
We can then rewrite $H_C$ as
\begin{align}
H_C 
    &= \mathbbm{1}C_{\text{avg}} -\sum_{\nu=1}^m \alpha_\nu H_\nu \\
    &= \mathbbm{1}C_{\text{avg}} -\sum_{\nu=1}^m \alpha_\nu Z_{\nu,1} Z_{\nu,2} Z_{\nu,3} \cdots Z_{\nu,t}.
\label{eq:k-local-hamiltonian-avg}
\end{align}
Using the above insights, we can form another theorem.

\begin{theorem}
Given a classical objective function $C(x)$ for a maximization task, represented by a $k$-local Hamiltonian $H_C = C_{\text{avg}}\mathbbm{1} - \sum_{\nu=1}^m \alpha_\nu H_\nu$, where each $H_\nu$ is a product of Pauli Z operators with at most $k$ terms, each $\alpha_\nu$ is a real number denoting the weight of the term, and $C_{\text{avg}}$ and $C_{\text{max}}$ denote the average and the global maximum of $C(x)$, if a QAOA protocol with $p$ rounds driven by the phase separator $H_1 = C_{\text{max}}\mathbbm{1} - H_C$ and the mixer $H_\text{Grover}$ that starts in the state $\ket{\psi_0} = \ket{+}^{\otimes n}$ reaches a state $\ket{\psi_p}$ with approximation ratio $\lambda$, then
\begin{align}
p
    &\geq \frac
    {1 - \left|\braket{\psi_0|\psi_{p}}\right|^2 + \lambda C_{\text{max}} - C_{\text{avg}}}
    {4 \pi \sqrt{\sum_{\nu = 1}^m \alpha_\nu^2}}.
\end{align}
\label{theorem:grover-local-cost}
\end{theorem}

\begin{proof}
Since we have already established that $C_{\text{avg}}$ is the coefficient of the trivial Pauli term, we now show that the standard deviation $\sigma_C$ can be easily calculated from the problem Hamiltonian representation as well.
Realizing that
\begin{align}
\sigma_C^2
    &= \braket{s|H_C^2|s} - \braket{s|H_C|s}^2\\
    &= \left(C_{\text{avg}}^2 + \sum_{\nu=1}^m \alpha_\nu^2\right) - C_{\text{avg}}^2 \label{eq:exp-square}\\
    &= \sum_{\nu=1}^m \alpha_\nu^2,
\end{align}
where the first term of \eqref{eq:exp-square} is due to the fact that only when all the Pauli $Z$'s cancel out into identity that we get nonzero contributions to the sum.
\end{proof}

This formulation of the $k$-local Hamiltonian not only allows us to calculate the statistical values of the objective functions easily, but also helps estimate the lower bound.
In some cases, $C_{\text{max}}$ can also be known when turning a search problem into an optimization problem for solving with the QAOA, e.g., turning SAT into Max-SAT~\cite{boulebnane-montanaro-sat-qaoa}.
One can also bound $C_{\text{max}} - C_{\text{avg}}$ from above using the sum of the absolute value of the coefficients of $H_C$, i.e., $C_{\text{max}} - C_{\text{avg}} \leq \sum_\nu | \alpha_\nu |$.

Similar arguments can be made for Hamiltonians of constrained problems; one only needs to adjust the initial state $\ket{s}$ from a uniform superposition of all bit strings to those in the feasible subspace. 
Although the adjusted state will not be a product state nor will the expectation value over nontrivial terms disappear, both the numerator and the denominator will still be polynomials in $n$.
For example, when the feasible bit string space is constrained by a fixed Hamming weight of $q$ (e.g., $q$-densest subgraph, $q$-vertex cover, etc.), each $t$-body term of the Pauli $Z$ can be calculated using sums of binomial coefficients.
Both the average and the variance can be calculated efficiently since one only needs to calculate up to $2k$-body terms and sum all the weighted coefficients of each term.

\subsection{Achieving constant approximation ratio requires a polynomial number of rounds for certain problems with Grover mixer} 
\label{sec:grover-mixer:approximation-ratio-unattainable}

There is an interesting property that one notices from this lower bound.
Unlike the case of the QAOA with the transverse-field mixer (which we will discuss in Sec.~\ref{sec:tf-mixer}), where a constant-round QAOA can give a guarantee on the approximation ratio $\lambda$ independent of the problem size for certain problems (e.g., Max-Cut and E3LIN2) as shown in~\cite{farhi-qaoa, farhi-e3lin2, wurtz-love-max-cut-constant-depth-qaoa}, this behavior does not exist for many classes of problems when solving with the QAOA and Grover mixer.

We can form a corollary for Max-Cut to formally state this insight.
\begin{corollary}
If a QAOA protocol with $p$ rounds finds an approximate solution with an approximation ratio $\lambda$ to the Max-Cut of a graph with $|E|$ edges driven by the objective value phase separator $H_1$ and the Grover mixer $H_\text{Grover}$ that starts in the state $\ket{\psi_0} = \ket{+}^{\otimes n}$, then
\begin{align}
    p &\geq \frac{
        1 - \left|\braket{\psi_0|\psi_{p}}\right|^2 + 
      \lambda C_{\text{max}} - |E|/2}
    { 2 \pi \sqrt{|E|} }.
\end{align}
\end{corollary}
\begin{proof}
We apply the lower bound of Theorem~\ref{theorem:grover-objective-bound} to an instance of Max-Cut with the number of edges denoted by $|E|$.
Since there are four possible ways to assign a cut partition to an edge, it is easy to see that the expected cut value for each edge is $\frac{1}{2}$, and since any edges of a pair are independent of each other, the expected cut values then equals $|E|/2$.
We can use a similar line of argument to calculate the standard deviation of the cuts of a graph, which evaluates to $\sqrt{|E|}/2$. (Refer to Appendix~\ref{appendix:max-cut-statistical-value} for full derivations of $\sigma_C$.)
\end{proof}

We know that for any bipartite graph, the maximum number of cut edges will be equal to the number of edges in the graph.
Using the above corollary, we get
\begin{align}
p &\geq \frac{
        1 - \left|\braket{\psi_0|\psi_{p}}\right|^2 + 
      \lambda |E| - |E|/2}
    { 2 \pi \sqrt{|E|} } \\
  &\geq \frac{(2\lambda - 1)}{4\pi}\sqrt{|E|}.
\end{align}
Since $|E|$ is at least linear in the number of vertices $n$ when the graph is connected, this shows that there exist problem instances in the class of Max-Cut that the QAOA with the Grover mixer cannot get any approximation ratio guarantee with constant rounds and would require the number of rounds to be at least a polynomial in $n$ to obtain a constant approximation ratio.
We note that this conclusion is different and cannot be compared directly to results in~\cite{bravyi-rqaoa}, where it is shown that the constant-round QAOA with the transverse-field mixer can obtain a constant approximation ratio for a certain family of regular bipartite expander graphs due to different assumptions and choice of mixer.

Although we have only shown this formally with the Max-Cut problems, similar calculations can be made for most combinatorial optimization problems.

\subsection{Tightness of the lower bounds for optimization problems with Grover mixer}
\label{sec:grover-mixer:tightness}

One natural question is whether or not the lower bounds derived here are tight.
We argue that the lower bounds on typical optimization problems would evaluate to a polynomial in $n$ using techniques presented here.
This is in contrast to the empirical results shown in~\cite{john-numerical-speedup-constrained, john-sat-qaoa} requiring an exponential number of rounds in $n$ and also go against the intuition obtained from the low-depth QAOA regime~\cite{mcclean-low-depth-quantum-optimization} that each round would only contribute to Grover-like progress at best.

The numerical results in~\cite{john-numerical-speedup-constrained, john-sat-qaoa} are both obtained by optimizing the QAOA via a round-by-round basis, using the angles found from \mbox{$p-1$} rounds to optimize for round $p$.
We argue that this method of round-by-round angle finding might not yield the globally optimal angles for the QAOA with $p$ rounds.
Another possibility that cannot be ruled out is that optimizing over expectation value might not lead us to the optimal angles. 
It is also infeasible to show this numerically since performing a grid search over $[0, 2\pi)^{2p}$ requires an exponential amount of computational resources for a high value of $p$.

As for the intuition given from the low-depth regime~\cite{mcclean-low-depth-quantum-optimization}, it was shown that the behavior from the first round of the QAOA with the Grover mixer can change the expectation value at most $C_{\text{max}}/\sqrt{N}$ similar to the progress made by one round of Grover's search algorithm. 
We point out that behaviors beyond low depth are hard to predict (this is true even in the case of a fixed angle selection~\cite{naphan-spo, koch-variational-amplitude-amplification}) since the phase of each computational basis state will not be the same across all states like in the first round.

Since our lower bound does not eliminate the usefulness of the QAOA with the Grover mixer, and with the above reasoning, the possibility of this type of QAOA to have a significant advantage over the Grover adaptive search~\cite{durr-hoyer-minimum-finding, gilliam-grover-adaptive-search-cpbo} remains an open question.

\section{Lower bounds on QAOA rounds with transverse-field mixer}
\label{sec:tf-mixer}

The transverse-field mixer is the go-to mixer for unconstrained problems and it is the first mixer to be studied in the QAOA literature.
It originates from QA due to being the most common and natural in most annealer devices.
Not only can the transverse-field mixer be implemented with constant depth on gate-based quantum computers, it has also been shown to significantly outperform other higher-order mixers, including the Grover-mixer in random satisfiability problems~\cite{john-sat-qaoa}.
The Hamiltonian of the transverse-field mixer can be defined as
\begin{align}
    H_\text{TF} = \mathbbm{1} \frac{n}{2} - \frac{1}{2} \sum_{j=1}^n X_j,
\end{align}
where $X_j$ denotes the Pauli $X$ acting on qubit~$j$.
We can now derive the lower bound on $p$ for the objective phase separator and transverse-field mixer.

\begin{lemma}
Given a classical objective function $C(x)$ for a maximization task, represented by the Hamiltonian $H_C$, encoded into the phase separator Hamiltonian $H_1 = \mathbbm{1}C_{\text{max}} - H_C$ and the mixing Hamiltonian $H_\text{TF}$ (transverse field), where all Hamiltonians are $2\pi$ periodic with zero ground-state energies,
and letting $C_{\text{max}}$ and $C_{\text{avg}}$ denote the global maximum and the average of $C(x)$,
if a QAOA protocol with $p$ rounds driven by $H_\text{TF}$ and $H_1$ that starts from the ground state $\ket{\psi_0} = \ket{+}^{\otimes n}$ of $H_\text{TF}$ reaches a state with approximation ratio $\lambda$, then 
\begin{align}
p \geq \frac
    {\frac{n}{2} - \frac{1}{2}\sum_{j=1}^{n} \braket{\psi_p| X_j |\psi_p} + \lambda C_{\text{max}} - C_{\text{avg}}}
    {4\pi \|[H_C, H_\text{TF}]\|}.
\end{align}
\label{lemma:tf-objective-bound}
\end{lemma}

\begin{proof}
This is an immediate consequence of Theorem~\ref{theorem:qaoa-objective-bound}.
\end{proof}

One immediate observation from this lower bound is the fact that energy from the term $\braket{H_\text{TF}}_p$ ranges between 0 to $n$, this is unlike the Grover-mixer case where the energy from the mixing Hamiltonian does not contribute significantly to the lower bound.
In the typical cases of optimization problems, we argue that $\lambda C_{\text{max}} - C_{\text{avg}}$ should dominate the numerator term.
When the energy with respect to the problem Hamiltonian $H_1$ is of order smaller than $n$, the lower bound will be dominated by the $\braket{H_\text{TF}}_p$ term.
This could come from optimization problems where the global optimum does not scale or grow sublinearly with the problem size.
We will explore the first case in this section and the latter case will be explored with the search cost function in Sec.~\ref{sec:search-lower-bound}.

We know from early results of the QAOA literature that constant-depth QAOAs achieve guarantees for approximation ratios $\lambda$ independent of the problem size for a few combinatorial optimization problems~\cite{farhi-qaoa, farhi-e3lin2, wauters-p-spin-qaoa, bravyi-rqaoa}. 
These combinatorial problems have one property in common, namely the number of interaction terms that each qubit participates in is bounded.
Our lower bounds should capture this behavior.
Thus, for the bounded-occurrence $k$-local Hamiltonians that correspond to these optimization problems with proven approximation ratio guarantees, our lower bound from Theorem~\ref{theorem:qaoa-objective-bound} should evaluate to a trivial lower bound for $p$ (i.e., less than one) that does not scale with the number of qubits or the number of cost terms.
We show in the next section a stronger result, namely, that our bounds evaluate to a trivial lower bound for all bounded-occurrence $k$-local Hamiltonian problems, thus leaving open the possibility that the QAOA with a transverse-field mixer might actually achieve constant approximation ratios for many additional optimization problems, but these guarantees have not yet been formally shown. 
Stated differently, our approach does not enable us to prove non-trivial lower bounds for bounded-occurrence $k$-local Hamiltonian problems.

While prior work has established a connection between the approximation ratio and the number of QAOA rounds (being constant)~\cite{bravyi-rqaoa}, these results are limited in scope, applying only to problem Hamiltonians with specific symmetries and a particular family of Ising models defined on expander graphs.
Moreover, the implications of their results differ from ours.
They demonstrate that constant-round QAOA performs worse than polynomial time classical algorithms for the problems they analyze, whereas we show that, without assuming any specific problem structure, the constant-round QAOA with the transverse-field mixer can still achieve constant approximation ratios.

\subsection{Bounded local problem Hamiltonians yield trivial lower bounds for QAOA with transverse-field mixer}
\label{sec:tf-mixer:bounded-occurence}

We assume that $H_C$ is a $k$-local Hamiltonian such that $k = O(1)$, where it can be expressed as a weighted sum of $m$ distinct products of Pauli $Z$ operators with spectral norm $\Theta(1)$ and each qubit participates in at most $l = O(1)$ terms.
We also assume that $C_{\text{max}} - C_{\text{avg}} = \Omega(n)$ making the contribution from $\braket{H_\text{TF}}_p$ negligible in the asymptotes.

Recalling from eq.~\eqref{eq:k-local-hamiltonian-avg} where the constant term in the problem Hamiltonian denotes $C_{\text{avg}}$ when the problem is unconstrained, we get 
\begin{gather}
    C_{\text{max}} \leq \|H_C\| = \|\mathbbm{1}C_{\text{avg}} + H_C'\| \leq C_{\text{avg}} + \|H_C'\|, \\
    C_{\text{max}} - C_{\text{avg}} \leq \|H_C'\| = O(m),
\label{eq:norm-hc-prime}
\end{gather}
where $H_C'$ is the nontrivial part of $H_C$ ($H_C$ minus multiples of identity operators).
Now if we can characterize the conditions when $\|[H_C, H_\text{TF}]\|$ is at least as large as $\|H_C'\|$, we will know when the lower bound evaluates to $p \geq 1$.

It is clear that the commutator $[H_C, H_\text{TF}]$ is also $k$-local but the number of terms can increase from $m$ terms up to $km$ terms and each qubit can participate up to $kl$ terms (originally up to $l$ terms).
We can express the commutator term as
\begin{align}
[H_C, H_\text{TF}]
    = -i \sum_{\nu=1}^m \alpha_\nu \sum_{l=1}^t Y_{\nu,l} \bigotimes_{j \neq l} Z_{\nu,j}.
\end{align}
Using Theorem 1 of Ref.~\cite{harrow-montanaro-eigenvalue-local-hamiltonian} (lower bound of the spectral norm of a traceless $k$-local Hamiltonian), we get 
\begin{align}
    \|[H_C, H_\text{TF}]\| = \Omega(km / \sqrt{kl}) = \Omega(m),
\end{align}
which holds when both $k$ and $l$ are constants independent of the number of qubits.
We can then conclude that 
\begin{align}
    p \geq \frac{C_{\text{max}} - C_{\text{avg}}}{4 \pi\|[H_C, H_\text{TF}]\|} = O(1),
\end{align}
and thus the lower bound becomes a trivial lower bound (i.e., less than one) when the problem Hamiltonian is a bounded-occurrence $k$-local Hamiltonian.

At first glance, this result may seem trivial, saying that $p$ needs to be larger than some constant.
However, recall that when $l$ is bounded by some constant, there are results, as mentioned earlier, showing guaranteed approximation ratios independent of problem sizes.
Hence, our bounds capture behaviors of the QAOA where the optimization problem is encoded into a bounded-occurrence $k$-local Hamiltonian.
The problems that are of this class are Max-Cut of bounded maximum degree, any constant bound regular graph, bounded degree $k$-SAT, or variations of odd degree Hamming weight problems (odd $p$-spin ferromagnets).
The trivial bound is obtained from our lower bounds because it captures easy problem instances of this type.

\subsection{Strictly $k$-local problem Hamiltonians yield trivial lower bounds for QAOA with transverse field}
\label{sec:tf-mixer:strictly-k-local}

In addition to the bounded-occurrence local problems, we can also make statements about strictly $k$-local Hamiltonians.
In order to show that this class of problem Hamiltonians also admits a trivial lower bound, we will construct a certain product state in which the expectation value with respect to the commutator is larger than the spectral norm of $H_C'$.

Since $H_C'$ is a classical Hamiltonian (diagonal in the computational basis), we know that there exists an eigenstate product state $\ket{s}$ which induced the norm and is from the Rayleigh quotient~\cite{horn-johnson-matrix-analysis} where
\begin{align}
    \braket{s|H_C'|s}^2 = \|H_C'\|^2.
\label{eq:rayleigh-quotient-technique}
\end{align}
Using the bit string $s$ of $\ket{s}$, we can construct $\ket{s^*} = \bigotimes_{j=1}^n \ket{s^*_j}$ by
\begin{align}
    \ket{s^*_j} = 
    \begin{cases}
        \cos(\theta) \ket{0} - i \sin(\theta) \ket{1}; \text{ if } s_j = 0, \\
        - i \sin(\theta) \ket{0} + \cos(\theta) \ket{1}; \text{ if } s_j = 1,
    \end{cases}  
\end{align}
where $\theta = -\frac{1}{2}\arctan(\sqrt{\frac{1}{k-1}})$.

With this construction, the following holds for every $j$ and $\nu$,
\begin{align}
    \braket{s^*|Y_j|s^*} &= \sqrt{\frac{1}{k}}\braket{s|Z_j|s}, \label{eq:inequality-ref}\\
    \braket{s^*|Z_j|s^*} &= \sqrt{\frac{k-1}{k}}\braket{s|Z_j|s}, \\
    \braket{s^*|[H_\nu, H_\text{TF}]|s^*} &= \braket{s|H_\nu|s}  \sqrt{k} \left(\sqrt{\frac{k-1}{k}}\right)^{k-1}.
\end{align}
This gives
\begin{align}
\|[H_C, H_\text{TF}]\| 
    &\geq |\braket{s^* | [H_C, H_\text{TF}] | s^*}| \\
    &= \|H_C'\| \sqrt{k} \left(\sqrt{\frac{k-1}{k}}\right)^{k-1} \\
    &\geq \|H_C'\| \text{ for all } k \geq 1,
\end{align}
where the inequality \eqref{eq:inequality-ref} is due to the Rayleigh quotient.
Recalling from eq.~\eqref{eq:norm-hc-prime} that the numerator term \mbox{$C_{\text{max}} - C_{\text{avg}}$} can only be as large as $\|H_C'\|$, we thus also get a trivial lower bound where the whole expression always evaluates to some value smaller than 1 from Theorem~\ref{theorem:qaoa-objective-bound} when the problem Hamiltonian is strictly $k$-local.

An example of a problem with a strictly 2-local cost Hamiltonian is the Max-Cut problem including the weighted graph variant where weights are integers (to preserve the periodicity condition).
The above result suggests that even without the bounded degree constraint (beyond regular graph families), there might be some classes other than the bounded degree of Max-Cut problems that can have a guarantee on the approximation ratio with a constant depth.

It is unclear how the lower bound would look for general $k$-local Hamiltonian problems.
However, we argue that the $\|[H_C, H_\text{TF}]\|$, in general, should be larger than $\|H_C\|$, from the intuition given in~\cite{harrow-montanaro-eigenvalue-local-hamiltonian}, stating that the degree or the number of terms a qubit participates in can decrease the norm at most on the order of a square root while increasing the number of terms should scale up the norm linearly. With this intuition, $\|[H_C, H_\text{TF}]\| \geq \sqrt{k}\|H_C\|$ should hold in the typical case.
We note that in order to get a non-trivial lower bound on $p$ with the transverse-field mixer, $\|[H_C, H_\text{TF}]\| = o(m)$ needs to hold (small $o$ notation).
We do not know of any upper bound of this order for optimization problems with local cost terms.

\section{Lower bound applied to search problems}
\label{sec:search-lower-bound}

We now consider search problems.
A search problem defined on $n$ qubits can be described as finding a bit string $z$ from a set $S = \{ z | C_\text{search}(z) = 1 \}$ where the objective function is defined as $C_\text{search}(x): \{0, 1\}^n \rightarrow \{0, 1\}$.

Although solving search problems with the QAOA is not practical in the hybrid quantum-classical algorithm sense when one needs to vary and find the angles, the minimum number of rounds to prepare a target state with a constant probability of observing a marked state is an important measure.
We know from the unstructured search lower bound results~\cite{bennett-strengths-and-weakness-of-quantum-computing, beals-lower-bounds-by-polynomials, zalka-grover-optimal} that in the black-box oracular model, searching requires $\Omega(\sqrt{N})$ queries to the oracle.
Since we consider only the QAOA, the model we consider is more constrained and the question becomes, given two parametrized unitaries (the adjustable phase oracle and the mixer), how many rounds one would need to alternately apply the two unitaries to prepare the target states. 

Since Grover's algorithm can also be viewed as the QAOA (the ansatz) because it alternates between two Hamiltonians albeit with fixed angles, our lower bound should be tight in the case of a search with the Grover mixer.
Not only do we know that Grover's algorithm~\cite{grover-alg} and its amplitude amplification generalization~\cite{brassard-bhmt-amplitude-amplification-estimation} are optimal in the unstructured search model, but a fixed angle selection strategy can also be made optimal (with constant success probability $\frac{1}{2}$) for the transverse-field mixer as well when there is only a single target state~\cite{jiang-transverse-field-search}.

We have shown in the previous two sections when the lower bound is applied to optimization problems with the Grover and the transverse-field mixer.
Although a compact expression (Theorem~\ref{theorem:grover-objective-bound}) and nontrivial lower bounds can be obtained for some problems for the QAOA with the Grover mixer, the tightness of the lower bound is still unclear.
On the other hand, we have given two conditions when the lower bound evaluates to a trivial lower bound for the transverse field case with local cost problems. 
In both conditions, it was assumed that the numerator term of Lemma~\ref{lemma:tf-objective-bound} is dominated by the contributions from the problem Hamiltonian $H_1$ ($\lambda C_{\text{max}} - C_{\text{avg}}$), while the case when the contributions from $H_0$ dominates have not been discussed.
Now we seek to give a partial answer to the tightness question regarding the Grover mixer and also address the case when $\braket{H_\text{TF}}$ dominates the expression in the transverse-field case by looking at search problems.

To be more precise, we show in Sec.~\ref{sec:search-lower-bound:grover-mixer} that the lower bound for the QAOA with search phase separator and Grover mixer is indeed tight leveraging Theorem~\ref{theorem:grover-objective-bound}, i.e., it matches the Grover lower-bound result of finding marked states in an unstructured search. 
We then focus on the QAOA for a search with the transverse-field mixer in Sec.~\ref{sec:search-lower-bound:tf-mixer}; we show that for two specific search problems, our lower bound yields two different sublinear and nonconstant lower bounds for the same number of marked states, thus showing that our proof approach leads to structure-dependent or problem-specific lower bounds for search problems. 
Finally, in Sec.~\ref{sec:search-overlap} we change our proof framework to work with state overlap expressions and in Sec.~\ref{sec:search-qaoa} recover the Grover lower bound for the QAOA with search and both Grover and transverse-field mixers (and in fact a much larger set of mixers).

\subsection{Lower bound for QAOA with Grover-mixer is tight for search}
\label{sec:search-lower-bound:grover-mixer}

We now show that the lower bound from Theorem~\ref{theorem:grover-objective-bound} is tight when applied to search problems.

\begin{corollary}
If a QAOA protocol solves a search problem with $p$ rounds, finding a marked state from $m$ marked states, with success probability $\lambda$, and is driven by the objective value phase separator $H_1$ (phase oracle with adjustable phase) and the Grover mixer $H_\text{Grover}$ that starts in the state $\ket{\psi_0} = \ket{+}^{\otimes n}$, then
\begin{align}
p   &\geq \frac{\lambda}{2\pi} \sqrt{\frac{N-m}{m}} - \frac{1}{2\pi} \sqrt{\lambda (1 - \lambda)}.
\label{eq:grover-search-bound}
\end{align}
\label{corollary:grover-search-bound}
\end{corollary}

\begin{proof}
We can see that, by rewriting the states as linear combinations of solution and nonsolution states, we get
\begin{gather}
    \left|\braket{\psi_0|\psi_p}\right|^2 = \frac{1}{N}\left( \sqrt{\lambda m} + \sqrt{(1 - \lambda)(N-m)} \right)^2, \label{eq:final-state-overlap-search}\\
    C_{\text{max}} = 1, \text{  } C_{\text{avg}} = m/N.
\end{gather}
The standard deviation for the search cost function is
\begin{align}
\sigma_{C} 
    &= \frac{\sqrt{m(N-m)}}{N}.
\label{eq:standard-deviation-search}
\end{align}
Plugging these into the lower bound in Theorem~\ref{theorem:grover-objective-bound} gives this expression, thus completing the proof.
\end{proof}

We can see that the expression in~\eqref{eq:grover-search-bound} recovers Grover's search scaling of order $\sqrt{N/m}$ from the lower bound in Theorem~\ref{theorem:grover-objective-bound} when the objective function is search and $m \ll N$.
This is consistent with the tight lower bound of an oracular unstructured search~\cite{brassard-bhmt-amplitude-amplification-estimation, bennett-strengths-and-weakness-of-quantum-computing, zalka-grover-optimal}.

\subsection{Search with transverse field QAOA}
\label{sec:search-lower-bound:tf-mixer}

We will now take a look at search problems with the transverse-field mixer.
Since the transverse-field mixer is not invariant under input permutation, we calculate the lower bound from Lemma~\ref{lemma:tf-objective-bound} with two constructed problems where the target states are marked differently.
(Proofs of the corollaries in this section are provided in Appendix~\ref{appendix:tf-search-derivations}.)

It is helpful to think about mixers as Laplacian of graphs~\cite{mcclean-low-depth-quantum-optimization}, and in this picture, the transverse-field mixer represents the $n$-dimensional hypercube graph.
Realizing this fact, we can think of the search set construction as marking vertices on the hypercube graph.

\subsubsection{Search on distance-3 independent set of hypercube}

First, consider a set $S_\text{dist-3}$ of marked states, where the Hamming distance between any two marked states is at least 3.
In other words, to get to another solution (if it exists) from a solution, one needs to perform at least 3 bit flips.
With this construction, we can derive the lower bound of the QAOA search with $H_\text{TF}$.

\begin{corollary}
Let $S_\text{dist-3}$ be the set of marked states of size $m$, where the Hamming distance between any two marked states is at least 3.
If a QAOA protocol solves a search problem with $p$ rounds, where marked states are defined by $S_\text{dist-3}$, with success probability $\lambda > \frac{1}{2}$ driven by the objective value phase separator $H_1$ (phase oracle with adjustable phase) and the transverse-field mixer $H_\text{TF}$ that starts in the state $\ket{\psi_0} = \ket{+}^{\otimes n}$, then
\begin{align}
p   &\geq \frac
    {n(1 - 2\sqrt{\lambda(1-\lambda)}) + 2\lambda - 2m/N}
    {4\pi \sqrt{n}}.
\end{align}
\label{corollary:tf-search-3-set-bound}
\end{corollary}

We can see that the lower bound above evaluates to around $p \geq \sqrt{n}$ independently of the number of marked states.
We defer the discussion of the above lower bound to the end of this section (Sec.~\ref{sec:search-lower-bound:tf-mixer:tightness}) to note the differences between the two  problems.

\subsubsection{Search for states with constant Hamming weight}

Now we construct another set of marked states, $S_\text{Hamming-k}$, where all the states having Hamming weight equal to $k$ are marked (or the Hamming distance from all 0 bit strings equal to $k$).

\begin{corollary}
Let $S_\text{Hamming-k}$ be the set of marked states of size $m$, where all states with Hamming weight $k$ are marked.
If a QAOA protocol solves a search problem with $p$ rounds, where marked states are defined by $S_\text{Hamming-k}$, with success probability $\lambda > \frac{1}{2}$ driven by the objective value phase separator $H_1$ (phase oracle with adjustable phase) and the transverse-field mixer $H_\text{TF}$ that starts in the state $\ket{\psi_0} = \ket{+}^{\otimes n}$, then
\begin{align}
    p &\geq \frac
    {n(1 - 2\sqrt{\lambda(1-\lambda)}) + 2\lambda - 2m/N}
    {4\pi \sqrt{2k(n-k) + n}}.
\end{align}
\label{corollary:tf-search-hamming-layer-bound}
\end{corollary}

We can see that the above lower bound depends on the choice of $k$.
For example, when $k = n / \log^2 n$, we get approximately $p \geq \log n$.

\subsubsection{Lower bound for search with the transverse-field mixer is structure dependent}
\label{sec:search-lower-bound:tf-mixer:tightness}

We have shown two examples of the lower bounds on $p$ when solving the search problem with the transverse-field mixer.
For the first problem with $S_\text{dist-3}$, we can choose the number of marked states from a single marked state up to around $2^n / n^2$ without affecting the lower bound ($p \geq \sqrt{n}$) since the lower bound holds for any marked states forming a distance-3 independent set on the hypercube graph. 
The number of marked states in the second problem with $S_\text{Hamming-k}$ is more restricted, but we can still vary $k$ and the number of marked states for chosen $k$ is $\binom{n}{k}$.

The two problems can be adjusted (e.g., choosing $k = n / \log^2 n$) to have the same number of marked states while the lower bounds evaluate to something of a different order ($\sqrt{n}$ and $\log n$).
Although the lower bounds seem to capture the problem structure without taking the number of marked states into account, this could merely be an artifact from Theorem~\ref{theorem:annealing-lower-bound} since we already know that the unstructured search lower bound is of order $\Omega(\sqrt{N/m})$ independent of the quantum search framework.
(We use the term \emph{structure} here loosely, referring to some patterns or some additional information regarding the input labeling and not just the knowledge of the distribution of the objective values since it is already well known that for quantum algorithms to gain significant speedup over specialized classical algorithms, it is required that the quantum algorithm must exploit certain structures of the problem~\cite{aaronson-ambainis-need-for-structure}.)
Thus, these two lower bounds are loose when applied to the transverse-field-mixer case.

\subsection{Lower bound on change in states overlap}
\label{sec:search-overlap}

To derive a better bound specific for search, we show that with a small adjustment to the lower bound given in Theorem~\ref{theorem:annealing-lower-bound}, we can derive another lower bound based on changes in the overlapping of states.
\begin{lemma}
Given two driving Hamiltonians with zero ground-state energies $H_0$ and $H_1$, where $H_1$ is $2\pi$ periodic, and letting $P_0$ be a projector that commutes with $H_0$,
for a QAOA protocol with $p$ rounds that starts in the ground state of $H_0$, it holds that
\begin{align}
    p \geq \frac{ \left| \braket{P_0}_{p} - \braket{P_0}_{0} \right| }{2\pi \left\| [P_0,H_1] \right\|}.
\end{align}
\label{lemma:probability-based-bound-annealing}
\end{lemma}

\begin{proof}
Consider dynamics with a Hamiltonian 
\begin{align}
    H(t) = (1-g(t))H_0 + g(t) H_1,
\end{align}
where the value of $g(t) = \{0, 1\}$ in the case of the QAOA. 
If $P_0$ is a projector that commutes with $H_0$, we have that 
\begin{align}
\left| \frac{d}{dt}\braket{P_0}_t \right| 
    &= \left| -i \tr{[H(t),\rho_t]P_0} \right| \\
    &= \left| \tr{\rho_t [P_0, g(t) H_1]} \right|\\
    &\leq |g(t)| \| [P_0,H_1] \| \label{eq:last-line-ref},
\end{align}
where \eqref{eq:last-line-ref} resulted from the cyclic property of trace and the Rayleigh quotient as done in eq.~\eqref{eq:rayleigh-quotient-technique}.
Integrating over the duration $t_f$ of the protocol gives
\begin{align}
\left| \braket{P_0}_{t_f} - \braket{P_0}_{0} \right| 
    &=  \left| \int_{0}^{t_f} \frac{d}{dt} \braket{P_0}_{t} dt \right| \\ 
    &\leq \| [P_0,H_1] \| \int_{0}^{t_f} |g(t)| dt \label{eq:probability-based-bounds-integration-line}\\ 
    &\leq \| [P_0,H_1] \| \sum_{j=1}^{p} |\gamma_j|\\ 
    &\leq \| [P_0,H_1] \| 2\pi p,
\end{align}
where we assumed $|\gamma_j| < 2\pi$ due to the periodicity.
Then 
\begin{align}
    p \geq \frac{ \left| \braket{P_0}_{p} - \braket{P_0}_{0} \right| }{2 \pi \left\| [P_0,H_1] \right\|},
\end{align}
completing the proof.
\end{proof}

As we can notice from this lower bound derived from the change in state overlap, the numerator is always between 0 and 1, and thus this bound would only be useful when the spectral norm of the commutator (the denominator) is small.
We will show that in the case of a search, this bound would prove to be much stronger and more general than our lower bound from Theorem~\ref{theorem:annealing-lower-bound}.

\subsection{Lower bound on search with QAOA}
\label{sec:search-qaoa}

The uniform superposition state $\ket{\psi_0} = \ket{+}^{\otimes n}$, the usual initial state for QA and the QAOA, is the ground state of the transverse-field Hamiltonian $H_\text{TF}$.
By choosing $P_0 = \ket{\psi_0}\bra{\psi_0}$ for Lemma~\ref{lemma:probability-based-bound-annealing}, we can form a theorem for the search.

\begin{theorem}
If a QAOA protocol solves a search problem with $p$ rounds, finding a marked state from $m$ marked states, with success probability $\lambda$ and is driven by the objective value phase separator $H_1$ (phase oracle with adjustable phase) and a mixing Hamiltonian $H_0$ that starts in the ground state $\ket{\psi_0} = \ket{+}^{\otimes n}$ of $H_0$, then
\begin{align}
p   &\geq \frac{\lambda(N-2m) + m - 2\sqrt{\lambda(1-\lambda)m(N-m)}}{2 \pi \sqrt{m(N-m)}}.
\end{align}
\label{theorem:search-bound-uniform-superposition-initial-state}
\end{theorem}

\begin{proof}
Considering the lower bound expression in Lemma~\ref{lemma:probability-based-bound-annealing}, when choosing $P_0 = \ket{\psi_0}\bra{\psi_0}$ we see that
\begin{align}
    [P_0, H_1] = -[\mathbbm{1} - P_0, H_1] = [H_1, H_\text{Grover}].
\end{align}
Then using Lemma~\ref{lemma:probability-based-bound-annealing}, we get
\begin{align}
p \geq \frac{ \left| \braket{P_0}_{p} - \braket{P_0}_{0} \right| }{2 \pi \left\| [P_0, H_1] \right\|} = \frac{1 - \left|\braket{\psi_0|\psi_p}\right|^2}{2\pi \sigma_C}. \label{eq:73}
\end{align}
By plugging the values from eqs.~\eqref{eq:final-state-overlap-search} and~\eqref{eq:standard-deviation-search} into the inequality \eqref{eq:73}, we get the expression as claimed, thus completing the proof.
\end{proof}

We note that the above lower bound can be adapted to other initial states, as long as the initial state is the ground state of the Hamiltonian $H_0$.
This includes all variants of amplitude amplifications with any diffusion operator (mixing Hamiltonian) that starts in the ground state of this operator~\cite{grover-alg, brassard-bhmt-amplitude-amplification-estimation} and with a slight modification to eq.~\eqref{eq:probability-based-bounds-integration-line} also applies to a continuous-time quantum-walk search~\cite{childs-quantum-walk-search, apers-quadratic-speedup-ctqw} that starts in the stationary distribution.

\section{Discussion}
\label{sec:discussion}

Although the QAOA is a very active topic of research in quantum algorithms and variational quantum algorithms, most results are of empirical nature, relying on numerical results of small problem instances and extrapolation to large instances while analytical results are few and far between.
It has been observed empirically that the choice of mixers can have a significant impact on the performance of the QAOA, in terms of both the number of rounds and the difficulty in the angle-finding classical loop.
In this work, we derived lower bounds in the number of QAOA rounds given a target constant approximation ratio when the driving Hamiltonians are periodic as Theorem~\ref{theorem:qaoa-objective-bound}.
We showed that the lower bounds derived also heavily depend on the spectral norm of the commutator of the phase separator and the mixer, strengthening the stance that QAOA research should invest more effort in understanding the relationship between the problem's structure and the mixer structure beyond the low-depth regime.
Despite the fact that our lower bounds can only be saturated by Grover's search algorithm, we consider it a first step toward the theoretical understanding of the QAOA protocols.

We also showed the limitations of our bounds when the problem Hamiltonian is bounded $k$-local or strictly $k$-local due to the spectral norm of the commutator of the problem and the transverse-field-mixing Hamiltonians are larger than the numerator given in our bounds resulting in a trivial fact that one needs at least one round.
This is because the bounds in~\cite{luis-pedro-annealing-lower-bounds} were derived from the changes of the energies with respect to the driving Hamiltonians.
Since the transverse field can flip the whole bit string, the energy can undergo a rapid change, thus rendering our lower bounds not so useful in these cases.
Possible ways to circumvent this problem are to characterize effective changes in the energy per one round of QAOA (under the small-angle regime~\cite{hadfield-analytical-framework-qaoa}),  
include more assumptions of angles at each round and utilize a tighter bound of~\cite{luis-pedro-annealing-lower-bounds}, or use other measures in addition to energies.

In the search problems with the transverse-field mixer, we saw that the lower bound captures certain structures while disregarding the number of marked states.
Although this could merely be an artifact of the lower bound derived from changes in energies, how the structure of the problem affects the bound is still an open problem.
We also showed that it is possible to derive problem-specific bounds like in the search case and certain bounds work better for certain problems.

\begin{acknowledgements}
We thank Yi\u{g}it Suba\c{s}i for helpful discussions.
Research presented in this article was supported by the Laboratory Directed Research and Development program of Los Alamos National Laboratory under Project No. 20230049DR.
LPGP acknowledges support from the DOE Office of Science, Office of Advanced Scientific Computing Research, Accelerated Research for Quantum Computing program, Fundamental Algorithmic Research for Quantum Computing (FAR-QC) project, and Beyond Moore’s Law project of the Advanced Simulation and Computing Program at LANL managed by Triad National Security, LLC, for the National Nuclear Security Administration of the U.S. DOE under Contract No. 89233218CNA000001.
\\ This work has been assigned the LANL technical report number LA-UR-23-29376.
\end{acknowledgements}

\bibliography{bib-zotero}

\onecolumngrid
\appendix

\section{Derivations of Grover-mixer and objective value phase separator lower bound}
\label{appendix:grover-obj-bound-derivations}

In the proof of Theorem~\ref{theorem:grover-objective-bound}, we gave an outline of how to calculate the spectral norm of the commutator term, for which we gave the result $\|[H_C, H_\text{Grover}]\| = \sigma_C$, where $\sigma_C$ is the standard deviation of the objective values.
We will show the derivations here.
Letting $\ket{\psi_0} = \ket{+}^{\otimes n}$, we get
\begin{align}
[H_C, H_\text{Grover}] 
    &= H_C H_\text{Grover} - H_\text{Grover} H_C \\
    &= -\left( 
     H_{C} \ket{\psi_0}\bra{\psi_0} - \ket{\psi_0}\bra{\psi_0} H_{C}
    \right) \\
    &= -\left( 
     (H_{C} \ket{\psi_0})\bra{\psi_0} - 
     \ket{\psi_0}(H_{C}^\dagger \ket{\psi_0})^\dagger
    \right) \\
    &= -\left( 
     \left(\frac{1}{\sqrt{N}} \sum_i H_{C} \ket{i}\right)\bra{\psi_0} - 
     \ket{\psi_0}\left(\frac{1}{\sqrt{N}} \sum_i H_{C} \ket{i}\right)^\dagger
    \right) \\
    &= -\frac{1}{\sqrt{N}} \sum_i C(i) \left( \ket{i}\bra{\psi_0} - \ket{\psi_0}\bra{i}  \right) \\
    &= -\frac{1}{N} \sum_i C(i) \left( \sqrt{N}\ket{i}\bra{\psi_0} - \sqrt{N}\ket{\psi_0}\bra{i}  \right) \\
    &= -\frac{1}{N} A,
\end{align}
where
\begin{align}
    A &= \left( \ket{C}\bra{e} - \ket{e}\bra{C} \right), \\
    \ket{C} &= \begin{bmatrix} C(0) & C(1) & C(2) & \cdots & C(N-1) \end{bmatrix}^T, \\
    \ket{e} &= \begin{bmatrix} 1 & 1 & 1 & \cdots & 1 \end{bmatrix}^T = \sqrt{N} \ket{\psi_0}. \\
\end{align}
Each entry of matrix $A$ can be described by
\begin{equation}
    A_{ij} = C(i) - C(j).
\end{equation}
We can see that $A$ is a skew-symmetric matrix, meaning that its eigenvalues are purely imaginary and come in pairs, and its spectral norm is equal to the absolute value of its largest eigenvalue.
Because $\ket{C}$ and $\ket{e}$ are not collinear, it means that the rank of $A$ is at most 2 and so the magnitudes of the two nonzero eigenvalues (if exist) will be equal.

There is a closed form for computing the spectral norm of a skew-symmetric matrix of this form~\cite{ipek-eigenvalue-rank-1-skew-symmetric} given by
\begin{equation}
    \|A\| = \sqrt{\beta N - \alpha^2},
\label{eq:spectral-norm-skew-symmetric}
\end{equation}
where 
\begin{align}
    \alpha =  \sum_i C(i) = N C_{\text{avg}} \text{ and }
    \beta = \sum_i C(i)^2.
\label{eq:spectral-norm-alpha-beta}
\end{align}
Using this formula, we get
\begin{align}
  \|[H_1, H_0]\|
    &= \frac{1}{N} \| A \| \\
    &= \frac{1}{N} \sqrt{N \sum_i C(i)^2 - \left(N C_{\text{avg}} \right) ^2}\\
    &= \sqrt{ \frac{\sum_i C(i)^2 - N {C_{\text{avg}}}^2}{N}} \\
    &= \sigma_{C},
\end{align}
where $\sigma_{C}$ is the standard deviation of the objective values.

\section{Statistical value of Max-Cut}
\label{appendix:max-cut-statistical-value}

We will show here the derivations regarding statistical values, the mean, and the standard deviation of cut values of Max-Cut problems.

\begin{lemma}
The average of cut values of a graph is $|E|/2.$
\end{lemma}
\begin{proof}
Since there are four possible ways to assign a cut partition to an edge, it is easy to see that the expected cut value for each edge is $\frac{1}{2}$, and since edges of a pair are independent of each other, the expected cut value then equals $|E|/2$.
\end{proof}

\begin{lemma}
The standard deviation of cut values of a graph is $\sqrt{|E|}/2.$
\end{lemma}

\begin{proof}
We define $c(U, V) = \{ e~|~e \in E; e_u \in U, e_v \in V\}$ \text{ where } $e = \{e_u, e_v\}$.
Letting $X$ be a random variable corresponding to the cut values,
we have

\begin{align}
2^n E[X^2] 
  &= \sum_{W \subseteq V} \left( \big| c(W, V \setminus W) \big| \right)^2 \\
  &= \sum_{W \subseteq V} \sum_{e \in c(W, V \setminus W)} \sum_{f \in c(W, V \setminus W)} \left( 1_e \cdot 1_f \right) \\
  &= \sum_{W \subseteq V} \left( \sum_{\substack{e,f \in c(W, V \setminus W), \\ \big| e \cap f \big| = 0}} (1_e \cdot 1_f) + \sum_{\substack{e,f \in c(W, V \setminus W), \\ \big| e \cap f \big| = 1}} (1_e \cdot 1_f) + \sum_{\substack{e,f \in c(W, V \setminus W), \\ \big| e \cap f \big| = 2}} (1_e \cdot 1_f) \right) \\
  &= \sum_{\substack{e,f \in E, \\ \big| e \cap f \big| = 0}} \sum_{\substack{W \subseteq V,\\ e_u \neq e_v, \\ f_u \neq f_v}} 1 +
        \sum_{\substack{e,f \in E, \\ \big| e \cap f \big| = 1}} \sum_{\substack{W \subseteq V,\\ e_u \neq e_v, \\ f_u \neq f_v}} 1 +
        \sum_{\substack{e,f \in E, \\ \big| e \cap f \big| = 2}} \sum_{\substack{W \subseteq V,\\ e_u \neq e_v, \\ f_u \neq f_v}} 1 \\
  &= \sum_{\substack{e,f \in E, \\ \big| e \cap f \big| = 0}} 4 \cdot 2^{n-4} +
        \sum_{\substack{e,f \in E, \\ \big| e \cap f \big| = 1}} 2 \cdot 2^{n-3} +
        \sum_{\substack{e,f \in E, \\ \big| e \cap f \big| = 2}} 2 \cdot 2^{n-2} \\
  &= 2^{n-2} \left( \sum_{\substack{e,f \in E, \\ \big| e \cap f \big| = 0}} 1 +
        \sum_{\substack{e,f \in E, \\ \big| e \cap f \big| = 1}} 1 +
        \sum_{\substack{e,f \in E, \\ \big| e \cap f \big| = 2}} 1 \right) + \sum_{\substack{e,f \in E, \\ \big| e \cap f \big| = 2}} 1 \\
  &= 2^{n-2} \left(\big| E \big|^2 + \big| E \big|\right).
\end{align}
Thus we get 
\begin{align}
\text{Var}(X) 
  &= E[X^2] - E[X]^2 \\
  &= \frac{\big| E \big|^2 + \big| E \big|}{4} - \left(\frac{\big|E\big|}{2}\right)^2 \\
  &= \frac{\big| E \big|}{4}, \\
\sigma(X) 
  &= \frac{\sqrt{\big| E \big|}}{2}.
\end{align}
This completes the proof.
\end{proof}

There is an alternate proof, which is far simpler, by writing the objective function of Max-Cut in the $k$-local Hamiltonian formalism.
We can write the Max-Cut Hamiltonian of an unweighted undirected graph as 
\begin{align}
    H_C = \frac{1}{2} \sum_{(u, v) \in E} \left(\mathbbm{1} - Z_u Z_v \right) = \frac{|E|}{2}\mathbbm{1} - \sum_{(u, v) \in E} \frac{1}{2} Z_u Z_v.
\end{align}
Using the technique outline in Sec.~\ref{sec:grover-mixer:k-local-hamiltonian}, we get
\begin{align}
\text{Var}(X)
  &= \sum_{e \in E} \left(\frac{1}{2} \right)^2 = \frac{|E|}{4},
\end{align}
which is the same as the calculations above.

\section{Derivations of lower bounds on QAOA rounds for search with transverse field}
\label{appendix:tf-search-derivations}

From Lemma~\ref{lemma:tf-objective-bound}, we have
\begin{align}
p \geq \frac
    {\frac{n}{2} - \frac{1}{2}\sum_{j=1}^{n} \braket{\psi_p| X_j |\psi_p} + \lambda C_{\text{max}} - C_{\text{avg}}}
    {4\pi \|[H_C, H_\text{TF}]\|}.
\end{align}
In the case of a search, we have
\begin{align}
    C_{\text{max}} = 1, \quad C_{\text{avg}} = \frac{|S|}{N},
\end{align}
where $S$ is the set of marked states and $N = 2^n$.

Let $H_S$ denote the Hamiltonian for a search problem given the set of marked states by $S$, $C(x) = \{0, 1\}$ denote the membership of $x$ in $S$, and $y \sim x$ denote that $y$ is Hamming distance 1 away from $x$.
Since we are considering searching with the transverse-field mixer, the commutator term can be written as
\begin{align}
[H_{S}, H_\text{TF}]
    &= H_S H_\text{TF} - H_\text{TF} H_S \\
    &= - H_S\left( \frac{1}{2}\sum_x \sum_{y \sim x}\ket{y}\bra{x}\right) + \left( \frac{1}{2}\sum_x \sum_{y \sim x}\ket{y}\bra{x}\right) H_S \\
    &= - \left( \frac{1}{2}\sum_x \sum_{y \sim x} C(y)\ket{y}\bra{x}\right) + \left( \frac{1}{2}\sum_x \sum_{y \sim x} C(x)\ket{y}\bra{x}\right) \\
    &= \frac{1}{2}\sum_x \sum_{y \sim x} (C(x) - C(y))\ket{y}\bra{x}.
\end{align}
Ignoring the $\frac{1}{2}$ prefactor, we can see that the commutator term can be viewed as a directed weighted subgraph of an $n$-dimensional hypercube with weight $\pm 1$.
If both vertices are marked or unmarked, there is no edge between them.
On the other hand, if they are in different sets, there are two edges between them; one edge from $x$ to $y$ and another with the opposite sign from $y$ to $x$.

\subsection{Search set from distance-3 independent set of hypercube}

We can now show the derivations for Corollary~\ref{corollary:tf-search-3-set-bound}.
Since no marked vertices are adjacent to each other nor are they sharing neighbors, the graph representing this commutator $[H_{S_\text{dist-3}}, H_\text{TF}]$ consists of $|S_\text{dist-3}|$ star graphs $K_{1,n}$.
Since the commutator is traceless and normal, the spectral norm equals the largest eigenvalue in magnitude, and since it also represents a graph, its spectral norm is the spectral radius of the corresponding graph.
The sign in the matrix does not matter in this particular case since all edges pointing from (or to) the center of the star share the same sign.
We then get (using Lemma~\ref{lemma:spectral-radius-of-star-graph})
\begin{align}
    \|[H_{S_\text{dist-3}}, H_\text{TF}]\| = \sqrt{n} / 2.
\end{align}

Next, we will show how to upper bound the numerator term.
First, we define a projector onto the subspace spanned by the marked states
\begin{align}
    P_S = \sum_{z \in S} \ket{z}\bra{z}.
\end{align}
We then rewrite the state $\ket{\psi_p}$ after $p$ rounds of the QAOA as
\begin{align}
    \ket{\psi_p} &= \sqrt{\lambda} \ket{\phi_0} + \alpha \ket{\phi_1} + \beta \ket{\phi_2}, \\
    \ket{\phi_0} &= \frac{P_S \ket{\psi_p}}{\|P_S \ket{\psi_p}\|}, \\
    \ket{\phi_1} &= \frac{\left( X_j P_S X_j \right) \ket{\psi_p}}{\|\left( X_j P_S X_j \right) \ket{\psi_p}\|}, \\
    \ket{\phi_2} &= \frac{\left(I - P_S - X_j P_S X_j \right)\ket{\psi_p}}{\|\left(I - P_S - X_j P_S X_j \right) \ket{\psi_p}\|}, 
\end{align}
where $\ket{\phi_0}$ denotes the projected state into the subspace spanned by the marked states, $\ket{\phi_1}$ denotes the projected state into the subspace spanned by the marked state having the $j\text{th}$ bit flipped, and $\ket{\phi_2}$ denotes projected state into the subspace orthogonal to the prior two subspaces.
We then have
\begin{align}
\braket{\psi_p | X_j | \psi_p} 
    &= 2 \text{Re}(\sqrt{\lambda} \alpha \braket{\phi_0|X_j|\phi_1}) + |\beta|^2 \braket{\phi_2|X_j|\phi_2} \\
    &\leq 2 \sqrt{\lambda} |\alpha| |\braket{\phi_0|X_j|\phi_1}| + |\beta|^2 \braket{\phi_2|X_j|\phi_2} \\
    &\leq 2 \sqrt{\lambda} |\alpha| + |\beta|^2  \\
    &\leq 2 \sqrt{\lambda} |\alpha| + 1 - \lambda - |\alpha|^2  \\
    &\leq 1 - (\sqrt{\lambda} - |\alpha|)^2  \\
    &\text{ (assuming that $\lambda > 1/2$, we get) } \\
    &\leq 1 - (\sqrt{\lambda} - \sqrt{1 - \lambda})^2  \\
    &\leq 2 \sqrt{\lambda(1-\lambda)}. 
\label{eq:tf-search-expectation-h0p}
\end{align}
This inequality holds true for all $j$; thus we have 
\begin{align}
    n - \sum_{j=1}^{n} \braket{\psi_p| X_j |\psi_p} \geq n \left(1 - 2\sqrt{\lambda(1-\lambda)}\right).
\end{align}

Substituting everything into Lemma~\ref{lemma:tf-objective-bound}, we get
\begin{align}
p &\geq \frac
    {\frac{n}{2} - \frac{1}{2}\sum_{j=1}^{n} \braket{\psi_p| X_j |\psi_p} + \lambda C_{\text{max}} - C_{\text{avg}}}
    {4\pi \|[H_C, H_\text{TF}]\|} \\
  &\geq \frac
    {\frac{1}{2} n \left(1 - 2\sqrt{\lambda(1-\lambda)}\right) + \lambda - m/N}
    {2\pi \sqrt{n}} \\
  &= \frac
    {n \left(1 - 2\sqrt{\lambda(1-\lambda)}\right) + 2\lambda - 2m/N}
    {4\pi \sqrt{n}}.
\end{align}

\subsection{Search set from bit strings with a fixed Hamming weight}

Now we will derive the proof for Corollary~\ref{corollary:tf-search-hamming-layer-bound}.
Consider the search set $S_\text{Hamming-k}$ where the marked states are selected from bit strings with Hamming weight equal to $k$.
The numerator term can be done similarly to that of marked states constructed from the distance-3 independent set of a hypercube.
For the commutator term $[H_{S_\text{Hamming-k}}, H_\text{TF}]$, using the same picture, we view it as an induced subgraph of a hypercube. 
Suppose we draw the graph $Q_n$ in levels, where all vertices in level $q$ are represented by a bit string with Hamming weight $q$; we get that there will only be edges connecting level $k-1$ and $k$, and $k$ and $k+1$.
Using the spectral radius property of this induced subgraph, we get (using Lemma~\ref{lemma:spectral-radius-of-3-layer-subgraph-hypercube})
\begin{align}
    \|[H_{S_\text{Hamming-k}}, H_\text{TF}]\| = \frac{\sqrt{2k (n-k) + n}}{2}.
\end{align}

Substituting everything into Lemma~\ref{lemma:tf-objective-bound}, we get
\begin{align}
p &\geq \frac
    {\frac{n}{2} - \frac{1}{2}\sum_{j=1}^{n} \braket{\psi_p| X_j |\psi_p} + \lambda C_{\text{max}} - C_{\text{avg}}}
    {4\pi \|[H_C, H_\text{TF}]\|} \\
  &\geq \frac
    {\frac{1}{2} n \left(1 - 2\sqrt{\lambda(1-\lambda)}\right) + \lambda - m/N}
    {2\pi \sqrt{2k (n-k) + n}} \\
  &= \frac
    {n \left(1 - 2\sqrt{\lambda(1-\lambda)}\right) + 2\lambda - 2m/N}
    {4\pi \sqrt{2k (n-k) + n}}.
\end{align}

\section{Spectral radius calculation of certain graphs}
\label{appendix:hypercube-spectral-radius}

\begin{lemma}
The spectral radius of a star graph $K_{1, n}$ is $\sqrt{n}$.
\label{lemma:spectral-radius-of-star-graph}
\end{lemma}

\begin{proof}
Suppose $A$ is the adjacency matrix of the star graph $K_{1, n}$.
We know that $A^2$ has a constant row sum equal to $n$, the number of ways to walk from the center vertex (any of the leaf vertices) to any of the leaf vertices (the center node) and back to center (any of the leaf vertices) in two steps; thus $||A|| = \sqrt{n}$ (cf. the Perron-Frobenius theorem~\cite{horn-johnson-matrix-analysis}).
\end{proof}

\begin{lemma}[spectral radius of subgraphs of $Q_n$ induced by vertices having Hamming weights $k-1$, $k$, and $k+1$]
For positive integers $k$ and $n$, where $0 < k < n/2$, and an $n$-dimensional hypercube graph $Q_n$, suppose $V_{k-1}$, $V_{k}$, and $V_{k+1}$ are three subsets of vertices of $Q_n$, where $V_m$ are the vertex set corresponding to vertices whose Hamming weight of the binary representation equals $m$, the spectral radius of the induced subgraph of $Q_n[V_{k-1} \cup V_{k} \cup V_{k+1}]$ is 
$\sqrt{2k (n-k) + n}$.
\label{lemma:spectral-radius-of-3-layer-subgraph-hypercube}
\end{lemma}

\begin{proof}    
We know that $Q_n[V_{k-1} \cup V_{k} \cup V_{k+1}]$ is a bipartite graph from the definition. 
Let $v_m \in V_m$.
We also know that $\text{deg}(v_k) = n$, $\text{deg}(v_{k-1}) = n-k+1$, and $\text{deg}(v_{k+1}) = k-1$.
Letting $A$ denote the adjacency matrix of the induced subgraph $Q_n[V_{k-1} \cup V_{k} \cup V_{k+1}]$, we can order the rows and the columns such that the first $|V_{k}|$ rows and columns correspond to $V_{k}$, the next $|V_{k-1}|$ rows correspond to $V_{k-1}$ and similarly for $V_{k+1}$.
We then see that $A^2$ produces a block matrix which we can separate into five blocks $A_{k,k}$, $A_{k-1,k-1}$, $A_{k-1,k+1}$, $A_{k+1,k-1}$, and $A_{k+1,k+1}$,
\begin{align}
A^2 = 
\begin{pmatrix}
  \begin{matrix}
  A_{k, k}
  \end{matrix}
  & \rvline &  \\
  \hline
  & \rvline &
  \begin{matrix}
  A_{k-1, k-1} & A_{k-1, k+1} \\
  A_{k+1, k-1} & A_{k+1, k+1}
  \end{matrix}
\end{pmatrix}.
\end{align}
In addition, we know that $\|A_{k, k}\| = 2nk - 2k^2 + n$ since it has a constant row sum, from the number of ways to walk from a vertex $u \in V_k$ and end up at $v \in V_k$ in two steps.
For the other block, since the Perron-Frobenius theorem says that the vector which induces the norm is positive, we get
\begin{align}
  \frac{x}{y} &= \frac{x (kn + k - k^2)  + y (n^2 - 2nk + n - k + k^2)}{x (k^2 + k)  + y (kn + n - k - k^2)} \\
              &= \frac{n - (k-1)}{k+1},
\end{align}
where $x$ and $y$ are elements of an eigenvector of length $\binom{n}{k-1} + \binom{n}{k+1}$ where the first $\binom{n}{k-1}$ entries are $x$ and the rest are $y$.
Then we get that the spectral norm of this block is
\begin{align}
  \left\|\begin{matrix}
  A_{k-1, k-1} & A_{k-1, k+1} \\
  A_{k+1, k-1} & A_{k+1, k+1}
  \end{matrix} \right\|
  &= kn + k - k^2 + \frac{y}{x}(n^2 -2nk + n - k + k^2) \\
  &= kn + k - k^2 + \frac{k+1}{n-k+1}(n^2 - 2nk + n - k + k^2) \\
  &= 2 k (n - k) + n.
\end{align}
Since the spectral norms of both diagonal block matrices are equal, this means that the spectral radius of the induced subgraph $Q_n[V_{k-1} \cup V_{k} \cup V_{k+1}]$ is exactly $\sqrt{2k (n-k) + n}$.
\end{proof}

\end{document}